\documentclass[11pt]{amsart}

\usepackage[vmargin=2cm, hmargin=2cm]{geometry}
\usepackage{enumerate}
\usepackage{url}
\usepackage{color}
\definecolor{red}{rgb}{1,0,0}
\definecolor{blue}{rgb}{0,0,1}

\newtheorem{theorem}{Theorem} 
\newtheorem{lemma}[theorem]{Lemma}
\newtheorem{corollary}[theorem]{Corollary}
\newtheorem{proposition}[theorem]{Proposition}

\newtheorem{algorithm}[theorem]{Algorithm}
\newtheorem{example}[theorem]{Example}
\newtheorem{remark}[theorem]{Remark}

\DeclareMathOperator{\Ap}{Ap}

\def\N{{\mathbb N}}
\newcommand\GAP{\textsf{GAP}}

\title[The second Feng-Rao number of inductive semigroups]{The second Feng-Rao number for codes coming from inductive semigroups}

\author{J. I. Farr\'{a}n}
\address{Departamento de Matem\'atica Aplicada, 
Universidad de Valladolid, Campus de Segovia, Espa\~na}
\email{jifarran@eii.uva.es}

\thanks{The first author is partially supported by the project MTM2012-36917-C03-01.}

\author{P. A. Garc\'{\i}a-S\'{a}nchez}
\address{Departamento de \'Algebra and CITIC-UGR, Universidad de Granada, 18071 Granada, Espa\~na}
\email{pedro@ugr.es}

\thanks{The second author is partially supported by the project MTM2010-15595, FQM-343,  FQM-5849 and FEDER funds.}



\begin{document}

\begin{abstract}
The second Feng-Rao number of every inductive numerical semigroup is explicitly computed. This number determines the asymptotical behaviour of the order bound for the second Hamming weight of one-point AG codes. In particular, this result is applied for the codes defined by asymptotically good towers of function fields whose Weierstrass semigroups are inductive. In addition, some properties of inductive numerical semigroups are studied, the involved Ap\'{e}ry sets are computed in a recursive way, and some tests to check whether a given numerical semigroups is inductive or not are provided. 
\end{abstract}

\maketitle

\section{Introduction}\label{sec:intro}

Garc\'{\i}a and Stichtenoth introduced in~\cite{GS} an asymptotically good tower of function fields, attaining the Drinfeld-Vl\u{a}du\c{t} bound given in~\cite{DV}. Such a tower allows to construct asymptotically good sequences of error-correcting codes, beyond the Gilbert-Varshamov bound (see~\cite{TsfVla}). The involved codes in this construction are one-point Algebraic Geometry codes (AG codes in short). These codes are based on algebraic curves $\chi$\/, and they are constructed by evaluating rational functions $\varphi$ with only one pole $Q$ at sufficiently many rational points $P_1,\ldots,P_n$ (see \cite{HvLP} and \cite{Stichtenoth} for a modern approach). Moreover, one-point AG codes can be efficiently decoded by the so-called Feng-Rao algorithm introduced in~\cite{FR}. This method corrects up to half the so-called Feng-Rao distance $\delta_{RF}(m)$ (also called {\em order bound}\/ in the literature), which can be obtained by numerical computations in the Weierstrass semigroup of $\chi$ at $Q$. 

In fact, the Feng-Rao distance improves the lower bound for the minimum distance given by the Riemann-Roch theorem, that is 
\[
\delta_{FR}(m)\geq m+1-2g
\]
for $m>2g-2$, where $m+1-2g$ is called the Goppa distance, $g$ being the genus of the underlying curve $\chi$\/, and $m$ the maximum pole order of a function $\varphi$ used to evaluate (see~\cite{HvLP} for further details). Moreover, the equality holds for $m>>0$ sufficiently large. 

Even though the Feng-Rao distance was introduced for Weierstrass semigroups and with decoding purposes, it is just a combinatorial concept that can be computed for arbitrary numerical semigroups, so that it can be computed just with numerical semigroup techniques like Ap\'{e}ry sets (see~\cite{WSPink}). The computation of Feng-Rao distances has been studied in the literature for different types of numerical semigroups (see~\cite{WSPink},~\cite{Arf},~\cite{Oneto} or~\cite{KirPel}). Later on, the concept of minimum distance for an error-correcting code has been generalized to the so-called {\em generalized Hamming weights} and the {\em weight hierarchy}\/.  These concepts were independently introduced by Helleseth et al. in~\cite{HKM} and Wei in~\cite{Wei} for applications in coding theory and cryptography, respectively. 

The Feng-Rao distance has been generalized in a natural way to higher weights (see~\cite{HeiPel}).  The obtained generalized Feng-Rao distances (or {\em generalized order bounds}), defined on the underlying numerical semigroup for an array of codes (or a weight function, in a modern setting), become again lower bounds for the corresponding generalized Hamming weights.  However, the computation of these generalized Feng-Rao distances is a much more hard problem than in the classical case. This means that very few results are known about this topic, and they are completely scattered in the literature (see for example~\cite{Bras},~\cite{fr-intervalos},~\cite{fr-IEEE},~\cite{JMDA},~\cite{WCC}, or~\cite{HeiPel}). 

This paper focuses on the asymptotical behaviour of the second Feng-Rao distance, that is, $\delta_{FR}^{2}(m)$ for $m>>0$ large enough. In fact, it was proved in~\cite{WCC} that
\[
\delta_{FR}^{r}(m)=m+1-2g+E_{r}
\]
and for $r=2$ in particular (details are made precise in the next section). The number $E_{r}\equiv \mathrm E(\Gamma,r)$ is called the $r$th Feng-Rao number of the semigroup $\Gamma$, and they are unknown but for very few semigroups and concrete $r$'s.  For example, it was proved in~\cite{fr-IEEE} that
\begin{equation}
\textrm E(S,r)=\rho_{r}
\end{equation}
for semigroups with only two generators. In~\cite{fr-intervalos} the authors also compute the Feng-Rao numbers for numerical semigroups generated by intervals. Note that the knowledge of $E_2$ provides a lower bound for $\delta_{FR}^{2}(m)$, namely 
\[
\delta_{FR}^{2}(m)\geq m+1-2g+E_{2}
\]
for $m\geq c$, $c$ being the conductor of the involved semigroup. In particular, we get a lower bound for the generalized Hamming weights in an array of codes whose associated semigroup is such a $\Gamma$. 
Our work is addressed to compute the second Feng-Rao number for inductive semigroups (see~\cite{PST} and \cite{PT}), by computing cardinalities of certain Ap\'{e}ry sets. 
This computation has an application to the tower of function fields introduced in~\cite{GS}. 

The paper is organized as follows. Section~\ref{sec:background} sets the general definitions concerning numerical semigroups, Feng-Rao distances, Feng-Rao numbers and inductive semigroups. The main Section \ref{sec:main} is devoted to calculate the cardinalities of Ap\'{e}ry sets for arbitrary inductive numerical semigroups and a explicit formula for the second Feng-Rao number for such a semigroup. As an application, we compute the second Feng-Rao number for every inductive semigroup involved in the tower of function fields given in~\cite{GS}, and furthermore we show how to compute the Ap\'{e}ry sets explicitly, not only their cardinalities, and the genus of an arbitrary inductive semigroup. Section \ref{sec:patterns} studies some extra properties of inductive semigroups, such as saturation, admissible patterns, and the embedding dimension. 
The paper ends with some examples and conclusions in Section~\ref{sec:examples_conclusions}.

\section{Inductive semigroups and Feng-Rao numbers}\label{sec:background}

This section presents some preliminary concepts on numerical semigroups, Feng-Rao distances and numbers, and inductive semigroups. 
We first recall the fundamentals of numerical semigroups. We will follow the notation in \cite{NS}.

Let $\Gamma$ be a numerical semigroup, that is, a submonoid of $(\N,+)$ 
with $\sharp(\N\setminus\Gamma)<\infty$ and $0\in\Gamma$ ($\mathbb N$ denotes the set of nonnegative integers). 
Denote by $g:=\sharp(\N\setminus\Gamma)$ the {\em genus} of $\Gamma$, and let $c\in \Gamma$ be 
the {\em conductor} of $\Gamma$\/, that is, the (unique) element in $\Gamma$ such that $c-1\notin\Gamma$ 
and $c+\mathbb N\subseteq \Gamma$. The elements of the set $\N\setminus\Gamma$ are called the {\em gaps} of $\Gamma$. The set $\Gamma\cap [0,c]$ is usually known as the \emph{set of small elements} (or sporadic elements).


It is well known (see for instance \cite[Lemma 2.14]{NS}) that $c\leq 2g$\/, and hence the \lq\lq last gap\rq\rq \ of $\Gamma$ 
is $c-1\leq 2g-1$. The number $c-1$ is called the {\em Frobenius number} of $\Gamma$\/. 
The {\em multiplicity} of $\Gamma$ is the least positive integer belonging to $\Gamma$. 

We say that a numerical semigroup $\Gamma$ is generated by a set of elements
$G\subseteq\Gamma$ if every element $x\in\Gamma$ can be written as a linear combination
\[
x=\displaystyle\sum_{g\in G}\lambda_{g}g,
\]
where $\lambda_{g}\in\N$ for all $g$ and only finitely many of them are non-zero.
In fact, it is classically known that every numerical semigroup
is finitely generated, so that we can always find a finite set $G$ generating $\Gamma$. Notice that we need at most one generator in each congruence class modulo the multiplicity of $\Gamma$.
Furthermore, every generating set contains the set of irreducible elements:
an integer $x\in\Gamma^*$ is \emph{irreducible} if whenever $x=u+v$ for some $u,v\in\Gamma$, we have that  $u\cdot v=0$ (as usual $\Gamma^*$ denotes $\Gamma\setminus \{0\}$). The set of irreducibles corresponds with the set $\Gamma^*\setminus (\Gamma^*+\Gamma^*)$, and it is indeed the unique minimal generating system of $\Gamma$. The cardinality of the minimal generating set of $\Gamma$ is  the {\em embedding dimension} of $\Gamma$ (more details in \cite{NS}).
One typically supposes that $\Gamma$ is minimally generated by $\{n_1<\cdots<n_e\}$, where $e$ is the embedding dimension (we use the notation $<$ in sets to denote that the elements are ordered in that way).  

If we enumerate the elements of $\Gamma$ in increasing order
\[
\Gamma=\{\rho_1=0<\rho_2<\cdots\},
\]
we note that every $x\geq c$ is the $(x+1-g)$th element of $\Gamma$\/,
that is $x=\rho_{x+1-g}\,$. With this notation, $\rho_2$ is the multiplicity of $\Gamma$. 

Finally, if $n\in\mathbb{Z}$ is any integer, we define the Ap\'{e}ry set of the semigroup $\Gamma$ related to $n$ as 
\[
\Ap(\Gamma,n) = \{ x\in\Gamma\;|\; x-n\notin\Gamma \}.
\]
It is known that $\sharp\Ap(\Gamma,n)=n$ if and only if $n\in\Gamma$ (see \cite[Proposition 1]{cyclotomic}). In this case, the set 
\[
(\Ap(\Gamma,n) \setminus \{0\}) \cup \{n\}
\]
is a generating system of $\Gamma$ with very nice properties (see for example \cite{NS}).  If $n$ is a gap of $\Gamma$, then $\sharp\Ap(\Gamma,n)>n$.

\subsection{Feng-Rao numbers}

Next we introduce the definitions of the generalized Feng-Rao distances, following the notations in \cite{fr-intervalos}.  Let $\Gamma$ be a numerical semigroup. 

\begin{enumerate}[(a)]

\item Given $x\in\Gamma$, we say that $\alpha\in\Gamma$ is a \emph{divisor} of $x$ if $x-\alpha\in\Gamma$ (in the literature, sometimes this fact is denoted by $a\le_\Gamma x$). Denote by $\mathrm D(x)=\{\alpha\in\Gamma\mid x-\alpha \in\Gamma\}$ the set of divisors of $x$. 

\item For $m_1\in\Gamma$, let $\nu(m_1):=\sharp \mathrm D(m_1)$. The (classical) {\em Feng-Rao distance} of $\Gamma$ is defined by the function
\[
\begin{array}{rcl}
\delta_{FR}\;:\;\Gamma&\longrightarrow&\N, \\
m&\mapsto&\delta_{FR}(m):=\min\{\nu(m_1)\mid m_1\geq m,\;\;m_1\in\Gamma\}.
\end{array}
\]
\end{enumerate}

There are some well-known facts about the functions $\nu$ and $\delta_{FR}$ for an arbitrary semigroup $\Gamma$
(see for example \cite{WSPink}, \cite{HvLP} or \cite{KirPel}).
An important result is that
$\delta_{FR}(m)\geq m+1-2g$ for all $m\in\Gamma$ with $m\geq c$\/,
and that equality holds if moreover $m\geq 2c-1$. 

In the sequel, we simplify the notation by writing $\delta(m)$ for $\delta_{FR}(m)$. 
The classical Feng-Rao distance corresponds to $r=1$ in the following definition.

Let $\Gamma$ be a numerical semigroup, $m\in\Gamma$ and $r\geq 1$. 
Set $\mathrm D(m_1,\ldots,m_r)=\mathrm D(m_1)\cup\cdots\cup \mathrm D(m_r)$. 
The \emph{$r$th Feng-Rao distance} of $\Gamma$ is defined by the function 
\[
\begin{array}{rcl}
{\delta^{r}}: \Gamma&\longrightarrow&\N, \\
m&\mapsto&\delta^{r}(m)=\min\left\{\sharp \mathrm D(m_1,\ldots,m_r) ~\middle|~ 
 m_1,\ldots,m_r\in \Gamma, \ m\le m_1<\cdots < m_r
 \right\}.
\end{array}\]

We know the asymptotic values of $\delta^r$ (see \cite{WCC}): there exists a certain constant $\mathrm E(\Gamma,r)$, 
depending on $r$ and $\Gamma$, such that 
\[
\delta^{r}(m)=m+1-2g+\mathrm E(\Gamma,r),
\]
for $m\geq 2c-1$. 
This constant $\mathrm E(\Gamma,r)$ is called the \emph{$r$th Feng-Rao number} of the semigroup $\Gamma$. 
In fact, it is also true that $\delta^{r}(m)\geq m+1-2g+\mathrm E(\Gamma,r)$ for $m\geq c$ (see \cite{WCC}). 

Feng-Rao numbers are known for numerical semigroups generated by intervals (\cite{fr-intervalos}) and for numerical semigroups with embedding dimension two (\cite{fr-IEEE}).

In this paper, we will focus on the second Feng-Rao number for inductive semigroups. It is easy to check from \cite[Section 4]{WCC} that 
\begin{equation}\label{calc-e2}
\mathrm E(\Gamma,2) = \min \{ \sharp\Ap(\Gamma,x) \mid 1\leq x\leq \rho_2 \}.
\end{equation}
Following \cite{WCC}, for the numerical semigroup $\Gamma$ and $x\in\mathbb{Z}$, we use the simplified notation
\[S_x=\Ap(\Gamma,x).\]

\subsection{Inductive semigroups}

We now recall the definition of inductive numerical semigroup. 

Let $a=(a_1,\ldots, a_n), b=(b_1,\ldots, b_n)\in \mathbb N^n$ with $a_i\ge 2$ and $b_{i+1}\ge a_ib_i$ for all $i\in\{1,\ldots, n\}$. Define 
\begin{itemize}
\item $\Gamma_0=\mathbb N$,
\item $\Gamma_i=a_i\Gamma_{i-1}\cup (a_ib_i+\mathbb N)$, for $i\in \{1,\ldots, n\}$.
\end{itemize}
We say that a numerical semigroup in inductive if it is of the form $\Gamma_n$. We will also write $\Gamma(a,b)=\Gamma_n$. In particular, an {\em ordinary} semigroup $S=\{0,c\rightarrow\}$ is inductive, with $n=1$, $a_{1}=c$ and $b_{1}=1$. 

For our purpose, we exclude the trivial case $\Gamma_0=\mathbb N$, so that we assume from now on that $n\geq 1$. For such an inductive numerical semigroup, set $\lambda_1=b_1$ and $\lambda_{i+1}=b_{i+1}-a_ib_i$ for $i\in\{2,\ldots,n-1\}$. Conversely, from the sequences $(a_1,\ldots, a_n)$ and $(\lambda_1,\ldots,\lambda_n)$ we can retrieve $b_{1}=\lambda_{1}$ and $b_{i+1}=\lambda_{i+1}+a_i b_i$. 

For $i\in \{1,\ldots,n\}$, define $A_i=\prod_{j=i}^n a_i$. Notice that $A_{1}$ is the multiplicity of $\Gamma_n$, and that $1<A_n<\cdots<A_1$. Denote by $I_{n}=\{1,\ldots,A_1\}$ the \lq\lq fundamental interval\rq\rq \ of $\Gamma_n$. Whenever it is necessary, we use the notation $A_{i}^{(k)}=\prod_{j=i}^k a_i$ with $i\leq k$ to specify that we are in the semigroup $\Gamma_k$ for $k\leq n$, instead of $\Gamma_n$. In this way, $A_{1}^{(k)}$ is precisely the multiplicity of $\Gamma_k$. When there is no possible misunderstanding, we will omit the superindex of $A_i^{(k)}$ and write simply $A_i$.   With the above notation one recovers the following known fact. 

\begin{lemma}\label{partition}
The numerical semigroup $\Gamma$ is a disjoint union of the following sets:
\begin{enumerate}[$\bullet$]
\item $\Lambda^1=\{0,A_1,2A_1, \ldots, \lambda_1 A_1\}$,
\item $\Lambda^2=b_1A_1+\{A_2,2A_2,\ldots,\lambda_2A_2\}$,
\item \dots
\item $\Lambda^{n}=b_{n-1}A_{n-1}+\{A_n,2A_n,\ldots,\lambda_nA_n\}$,
\item $\Lambda^{n+1}=(a_nb_n+1)+\mathbb N$. 
\end{enumerate}
\end{lemma}
\begin{proof}
It follows easily just observing that $\sum_{i=1}^i \lambda_j A_j=b_iA_i$.
\end{proof}

Note that, from the above result, it is clear that every inductive semigroup is acute, in the sense of \cite{Oneto}. In fact, every Arf semigroup is known to be acute (see~\cite{acute}), and inductive semigroups have the Arf property (see~\cite{Arf}). As we will see later, the main interest of these semigroups is that they appear in a natural way as the Weierstrass semigroups associated to certain towers of function fields (see \cite{GS}).

\subsubsection{Testing inductiveness}

Let $\Gamma$ be a numerical semigroup. The following naive procedure can be used to check if it is inductive.

Take  $S$ to be the small elements of $\Gamma$, and $a=\gcd(S)$. If $a=1$, then $\Gamma$ is inductive if and only if $\Gamma=\mathbb N$. Otherwise, let $\Lambda$ be the semigroup $S/a\cup (\max\{S/a\}+\mathbb N$). It follows that $\Gamma$ is inductive if and only if $\Lambda$ is inductive. 

Another way to proceed is to check if we can decompose $\Gamma$ as in Lemma \ref{partition} for suitable $A_n\mid A_{n-1}\mid \cdots \mid A_1$ (and suitable $b_1,\ldots, b_n$). In fact, we have the following characterization. 

\begin{proposition}\label{characterization}

For a numerical semigroup $\Gamma$, consider $\delta_i:=\rho_{i+1}-\rho_{i}$ ($i\geq 1$) the distance between two consecutive elements in $\Gamma$. Then $\Gamma$ is inductive if and only if $\delta_{i+1}\mid \delta_i$ for all $i\geq 1$. 

\end{proposition}

\begin{proof}

If $\Gamma$ is inductive, we apply Lemma \ref{partition}, and set $A_{n+1}:=1$. Thus, the distance between any two consecutive elements inside a set $\Lambda^k$ is $A_k$, the distance between the last element of $\Lambda^k$ and the first one in $\Lambda^{k+1}$ ($1\leq k\leq n$) is $A_{k+1}$, and $A_{k+1}\mid A_k$ for all $1\leq k\leq n$, and the condition of the statement is satisfied. 

Conversely, assume that $\Gamma$ satisfies the condition $\delta_{i+1}\mid \delta_i$ for all $i\geq 1$. Then, suppose that the first $\lambda_1$ distances $\delta_i$ are equal to some number $A_1$, the next $\lambda_2$ distances are equal to $A_2$, and so on, until $\lambda_n$ distances are equal to $A_n$, and from there on all the distances are $A_{n+1}:=1$. Now, since by assumption $A_{k+1} \mid A_k$ ($1\leq k\leq n$), define $a_k:=A_k / A_{k+1}$, $b_1:=\lambda_1$, and $b_k := \lambda_k + a_{k-1}b_{k-1}$ for $2\leq k\leq n$, and it follows that $\Gamma$ is inductive. 
\end{proof}

Notice that we only need to check up to the conductor, because from there on the deltas are equal to one, and the procedure terminates after a finite number of tests. 


\section[The second Feng-Rao number]{The second Feng-Rao number of an inductive semigroup}\label{sec:main}

Our purpose now is to compute the second Feng-Rao number of inductive numerical semigroups. To that end, we need some technical results. 

\begin{lemma}\label{linearity}
Let $j\in \{1,\ldots, a_n-1\}$ and $k\in \{0,\ldots, A_1/a_n-1\}$. 
Then $\sharp S_{ka_n+j}=\sharp S_{ka_n+1}+(j-1)$.
\end{lemma}
\begin{proof}
Since $S_{ka_n+1}=\{x\in\Gamma \;|\; x-(ka_n+1) \notin \Gamma\}$ and $j<a_{n}$, that is the length of the largest interval of gaps of $\Gamma$, one has $S_{ka_n+1}\subseteq S_{ka_n+j}$. The remaining elements in $S_{ka_n+j}\setminus S_{ka_n+1}$ are  $c+ka_n+1,\ldots,c+ka_n+j-1$, with $c=a_n b_n$ the conductor of $\Gamma$. With this, the formula is proved. 
\end{proof}

\begin{lemma}\label{monotony}
Under the assumptions of the above lemma, if moreover $k>0$ one has $\sharp S_{ka_n}\le \sharp S_{ka_n+1}$.
\end{lemma}

\begin{proof}
We make use of the disjoint decomposition $\Lambda^m$ of $\Gamma$. In fact, since all the elements of $\Lambda^{1} \cup \cdots \cup \Lambda^{n}$ are multiples of $a_{n}$, then 
\[
S_{ka_n+1} \cap (\Lambda^{1} \cup \cdots \cup \Lambda^{n}) = \Lambda^{1} \cup \cdots \cup \Lambda^{n}.
\]
Trivially, $S_{ka_n} \cap (\Lambda^{1} \cup \cdots \cup \Lambda^{n}) \subseteq \Lambda^{1} \cup \cdots \cup \Lambda^{n}$. 

Observe that 
\[
1 + (S_{ka_n}\cap\Lambda^{n+1}) \subseteq S_{ka_n+1}\cap\Lambda^{n+1}.
\]
and the inclusion is strict if and only if $a_n b_n -ka_n \notin \Gamma_n$. 
Combining both facts one obtains the desired inequality. 
\end{proof}

As a consequence of the previous results, the function $\sharp S_{x}$, for $x\in I_n$, is increasing with respect to $x$ except for (possibly) the elements of the form $x=ka_n$, where this function can probably drop. Thus, only these elements, together with $x=1$, must be taken into account in order to compute the minimum 
\[
\mathrm E(\Gamma,2)=\min\big\{\sharp S_{x} \mid x\in\{1,\ldots,A_{1}\}\big\}.
\]

\begin{lemma}\label{induction}
Let $\Gamma$ be a numerical semigroup with conductor $c$, and let 
\[
\Gamma' := a\cdot\Gamma \cup (ab+\mathbb N),
\]
with $b$ an integer, $b\geq c$. Then for $t\geq 1$ one has 
\[\Ap(\Gamma',at)=a\Ap(\Gamma,t)\cup (\{a b, a b+1,\ldots,a b+at-1\}\setminus \{ ab,a b+a,a b+2a,\ldots, a b+(t-1)a\}),\]
and this union is disjoint.
In particular, 
\[
\sharp \Ap(\Gamma',at) = \sharp \Ap(\Gamma,t) + (a-1)t.
\]
\end{lemma}
\begin{proof}
If we take an element $\lambda'\in\Gamma'$ and $\lambda'-at\notin\Gamma'$, in particular $\lambda'-at<a b$. Then $\Ap(\Gamma',at)$ can be decomposed into two disjoint subsets, depending on whether $\lambda'$ is multiple of $a$ or not, as follows 
\[
\Ap(\Gamma',at) = \{\lambda'\in\Gamma'\setminus a\cdot\mathbb{N}\mid \lambda'-at<a b\} \cup \{a\lambda \mid a\lambda\in\Gamma' \mbox{ and } a\lambda-at\notin\Gamma'\}.
\]
The elements in $\Gamma'\setminus a\cdot \mathbb N$ are all in $[a b,\infty)\setminus a\cdot \mathbb N$. Hence
\[
\{\lambda'\in\Gamma'\setminus a\cdot\mathbb{N}\;|\;\lambda'-at<a b\} = \{ab,a b+1,\ldots,a b+at-1\}\setminus\{a b, a b+a,a b+2a,\ldots, a b+(t-1)a\}. 
\]
It remains to prove that the set $\{a\lambda \mid a\lambda\in \Gamma'\hbox{ and }a\lambda-at\not\in \Gamma'\}$ equals $a\cdot \mathrm{Ap}(\Gamma,t)$. 
\begin{itemize}
\item[$\supseteq$] If $\lambda\in \Ap(\Gamma,t)$, then $\lambda-t\not\in \Gamma$. Hence $\lambda-t< c\le b$, and consequently $a\lambda -at<ab$. This forces $a\lambda -at\not\in a\Gamma$ and $a\lambda -at\not\in (ab+\mathbb N)$, whence $a\lambda-at\not\in\Gamma'$.

\item[$\subseteq$] Let $a\lambda\in \Gamma'$ be such that $a\lambda -at\not\in \Gamma'$. This in particular means that $a\lambda -at\not\in a\Gamma$. Also $a\lambda \in \Gamma'$ and $b\ge c$ imply that $\lambda \in \Gamma$. Hence $\lambda\in \Gamma$ and $\lambda -t\not\in \Gamma$, which yields $a\lambda \in a\Ap(\Gamma,t)$. \qedhere
\end{itemize}

\end{proof}

Note that the term $(a-1)t$ in the above formula is increasing with respect to $t$. This allows us to obtain the following result. 

\begin{proposition}\label{multiples}
Under the standing hypothesis, for a fixed $i\in\{2,\ldots, n\}$ we have 
\[
\min\{\sharp S_x \mid x\in\{A_i,A_{i}+1,\ldots,A_{i-1}-1\}\}=\sharp S_{A_i}.
\]
\end{proposition}

\begin{proof}

We proceed by induction on $n\geq 2$ (for $n=1$ there is nothing to prove). 
Note that, because of Lemmas \ref{linearity} and \ref{monotony}, it suffices to prove that $\sharp S_{x} \geq \sharp S_{A_{i}}$, 
for $x=A_{i} + ka_{n}$ and $0\leq k < (a_{i-1}-1)\cdot a_{i}\cdots a_{n-1}$, that is, $x<A_{i-1}$. 

Thus, for $n=2$, let $x=A_2+ka_{2}=a_2 + k a_2$, with $0\leq k<a_1 -1$. In other words, $x=j a_2$ for $j\in \{1,\ldots,a_1 -1\}$. 
In this case, we apply Lemma \ref{linearity} to $\Gamma_1$, $k=0$ and $j\in\{1,\ldots,a_1 -1\}$, so that $\sharp\Ap(\Gamma_1 ,j)$ is linearly increasing in $j$, for such values of $j$. Now we use Lemma \ref{induction} with $\Gamma=\Gamma_1$, $a=a_2$, and $b=b_2$, taking $t=j$, so that the inequality $\sharp S_{ta_2} \geq \sharp S_{a_2}$ holds, and the statement follows easily. 

Now assume that the statement holds for $n$ and let us prove it for $n+1$. By Lemma \ref{induction}, and using that $A_i^{(n+1)}/a_{n+1}=A_i^{(n)}$, we obtain
\[\sharp \Ap(\Gamma_{n+1}, A_i^{(n+1)}+ka_{n+1}) = \sharp \Ap(\Gamma_n, A_i^{(n)}+k)+(a_{n+1}-1)(A_i^{(n)}+k),\] for all suitable $k$. Now the induction hypothesis is telling us that $\sharp\Ap(\Gamma_n,A_i/a_{n+1})\le \sharp \Ap(\Gamma_n, A_i/a_{n+1}+k)$ for every nonnegative integer $k$ such that $A_i^{(n)}+k <A_{i-1}^{(n)}$, or equivalently, $A_i^{(n+1)}+ka_{n+1}<A_{i-1}^{(n+1)}$. Hence
\begin{multline*}
\sharp \Ap(\Gamma_{n+1}, A_i^{(n+1)}) =\sharp\Ap(\Gamma_n,A_i^{(n)})+(a_{n+1}-1)A_i^{(n)}\\ \le 
\sharp \Ap(\Gamma_n, A_i^{(n)}+k)+(a_{n+1}-1)(A_i^{(n)}+k) = \sharp \Ap(\Gamma_{n+1}, A_i^{(n+1)}+ka_{n+1}).
\end{multline*}
With the use of Lemmas \ref{linearity} and \ref{monotony} we conclude the proof.
\end{proof}

As a consequence of the above results, one obtains the following result. 

\begin{theorem}\label{E2}
For an inductive numerical semigroup $\Gamma=\Gamma(a,b)$, with $a\in \mathbb N^n$,
one has
\[
\mathrm E(\Gamma,2)=\min
\{
\sharp S_{1},\sharp S_{A_{n}},\sharp S_{A_{n-1}},
\ldots,\sharp S_{A_{2}},\sharp S_{A_1}
\}.
\]
\end{theorem}

We now compute explicitly the numbers involved in the above theorem. 

\begin{theorem}\label{SAi}
For an inductive numerical semigroup $\Gamma_{n}$ with $n\geq 1$ one has
\[
\sharp S_{1} = \lambda_{1}+\cdots+\lambda_{n} + 1,
\]
where $\lambda_{1}=b_{1}$ and 
\[
\sharp S_{A_{n-k}} = \lambda_{1}+\cdots+\lambda_{n-k-1} + A_{n-k}
\] 
for $k\in \{0,\ldots, n-1\}$. 
\end{theorem}

\begin{proof}

The formula for $\sharp S_{1}$ is obvious for all $n\geq 1$, since by construction it provides the number of deserts. 

Take then $k\in \{0,\ldots, n-1\}$. We proceed by induction in $n$. 

\noindent$\boldsymbol{n=1}$. We only have to compute $\sharp S_{A_1}$, where $A_{1}=a_{1}$ is the multiplicity of $\Gamma_1$. In fact 
\[
\Ap(\Gamma_1,a_1)=\{0,a_1 b_1 +1,\ldots,a_1 b_1 +(a_1 -1)\}
\]
and hence $\sharp S_{A_1}=A_{1}$. 

\noindent$\boldsymbol{n}$ implies $\boldsymbol{n+1}$. Assume that for $\Gamma=\Gamma_n$ the statement is true. Take $\Gamma'=\Gamma_{n+1}$ and apply again Lemma \ref{induction} with $a=a_{n+1}$ and $b=b_{n+1}$ as follows. 

\noindent $\sharp S_{A_{n+1}}$: taking $t=1$, we get 
\[
\sharp \Ap(\Gamma_{n+1},a_{n+1}) = \sharp \Ap(\Gamma_n,1) + (a_{n+1} -1) = \lambda_1+\cdots+\lambda_n + 1 + a_{n+1} -1 = \lambda_1+\cdots+\lambda_n + A_{n+1}.
\]

\noindent $\sharp S_{A_{n+1-k}}$: ($k\in\{1,\ldots,n\}$) set $t=A_{n-(k-1)}^{(n)}=a_{n+1-k}\cdots a_{n}$; we obtain 
\begin{multline*}
\sharp \Ap(\Gamma_{n+1},A_{n+1-k}) = \sharp \Ap(\Gamma_n,A_{n-(k-1)}^{(n)}) + (a_{n+1} -1)A_{n-(k-1)}^{(n)} \\
=\lambda_1+\cdots+\lambda_{n-k} + A_{n-(k-1)}^{(n)} + (a_{n+1} -1)A_{n-(k-1)}^{(n)} \\ 
=\lambda_1+\cdots+\lambda_{n-k} + A_{n-(k-1)}^{(n)} a_{n+1} = \lambda_1+\cdots+\lambda_{n-k} + A_{n+1-k}.\qedhere
\end{multline*}

\end{proof}

Notice that the equality $\sharp S_{A_{1}}=A_1$ follows also as a consequence of $A_1 \in \Gamma_n$. 
It is clear from the above result that, depending on the choice of the numbers $a_{i}$ and $\lambda_{j}$, the minimum can be achieved in any of the involved $\sharp S_{x}$, so that we cannot skip any element in this formula. We illustrate this with an example.

\begin{example}\label{sharpness}

Consider $n=4$. In order to compute the second Feng-Rao number of $\Gamma_4$ we have to take the minimum from the following integers:

\begin{enumerate}[(1)]

\item $\lambda_1 + \lambda_2 + \lambda_3 + \lambda_4 +1$, 

\item $\lambda_1 + \lambda_2 + \lambda_3 + a_4$, 

\item $\lambda_1 + \lambda_2 + a_3 a_4$, 

\item $\lambda_1 + a_2 a_3 a_4$, 

\item $a_1 a_2 a_3 a_4$. 

\end{enumerate}
Each of these values can be achieved as the minimum, for instance: 
\begin{enumerate}[$\bullet$]

\item if $a_1=a_2=a_3=a_4=4$ and $\lambda_1=\lambda_2=\lambda_3=\lambda_4=2$, then the minimum is the value (1), 

\item if $a_1=a_2=a_3=3$, $a_4=2$ and $\lambda_1=\lambda_2=\lambda_3=\lambda_4=2$, then the minimum is the value (2), 

\item if $a_1=a_2=4$, $a_3=a_4=2$, and $\lambda_1=\lambda_2=\lambda_3=\lambda_4=4$, then the minimum is the value (3), 

\item if $a_1=a_2=a_3=a_4=2$ and $\lambda_1=\lambda_2=\lambda_3=\lambda_4=6$, then the minimum is the value (4), 

\item if $a_1=a_2=a_3=a_4=2$ and $\lambda_1=\lambda_2=\lambda_3=\lambda_4=9$, then the minimum is the value (5). 

\end{enumerate}

\end{example}

As a consequence of Theorems \ref{E2} and \ref{SAi}, one can easily compute the second Feng-Rao number $\mathrm E = \mathrm E(\Gamma,2)$ for an inductive numerical semigroup $\Gamma$ by using the following algorithm. 

\begin{algorithm}\label{computing}
{\rm Input:} ${\bf a},{\bf b}$ of length $n$ 

Compute $\lambda_1,\ldots,\lambda_n$ 

Compute $L:=\lambda_1+\cdots+\lambda_n$ 

Set $A:=1$ 

Set $\mathrm E:=L+A$ 

for i in $\{1,\ldots,n\}$ do 

\hspace{1.cm} $L:=L-\lambda_{n+1-i}$ 

\hspace{1.cm} $A:=A\, a_{n+1-i}$ 

\hspace{1.cm} if $L+A<\mathrm E$ then $\mathrm E:=L+A$ end 

end

{\rm Output:} $\mathrm E$ 

\end{algorithm}

\subsection{More on Ap\'{e}ry sets}

From the discussions above,  we are able to compute explicitly the Ap\'{e}ry sets of an inductive numerical semigroup with respect an element $x$ in the fundamental interval, by means of a recursive procedure.

Thus, for $n=1$ consider the semigroup $\Gamma_1 = a_1\mathbb{N} \cup (a_1 b_1 + \mathbb{N})$ and $x\in\{1,\ldots,a_1\}$. According to Lemma \ref{partition}, this semigroup is partitioned into the following two sets: 
\begin{itemize}
	\item $\Lambda^1=\{0,a_1,2a_1, \ldots, \lambda_1 a_1\}$, 
	\item $\Lambda^{2}=(a_1 b_1+1)+\mathbb N$. 
\end{itemize}
Then we obviously have the following result. 

\begin{proposition}\label{Apery-1}
	Let $a_1$ and $b_1$ be integers greater than one. Set $\Gamma_1 = a_1\mathbb{N} \cup (a_1 b_1 + \mathbb{N})$ and let $x\in\{1,\ldots,a_1\}$.
	
	\begin{enumerate}[(1)]
		
		\item $\Ap(\Gamma_1,1)=\Lambda^1$. 
		
		\item $\Ap(\Gamma_1,x)=\Lambda^1 \cup \{a_1 b_1+1,\ldots,a_1 b_1+x-1\}$ for $1<x<a_1$. 
		
		\item $\Ap(\Gamma_1,a_1)=\{0\} \cup \{a_1 b_1 +1,\ldots,a_1 b_1 +(a_1 -1)\}$. 
		
	\end{enumerate}
	
\end{proposition}

Now for $n>1$ we can obtain the Ap\'{e}ry sets of $\Gamma_n$ in terms of those of $\Gamma_{n-1}$ as follows.

\begin{proposition}\label{Apery-n}
	Let $\Gamma_n=\Gamma(a,b)$ with $a=(a_1,\ldots, a_n),b=(b_1,\ldots, b_n)\in \mathbb N^n$, $a_i\ge 2$ and $b_{i+1}\ge a_ib_i$ for all $i$.
	
	\begin{enumerate}[(1)]
		
		\item $\Ap(\Gamma_n,1)=\Lambda^1 \cup \cdots \cup \Lambda^n$. 
		
		\item $\Ap(\Gamma_n,x)=\Lambda^1 \cup \cdots \cup \Lambda^n \cup \{a_1 b_1+1,\ldots,a_1 b_1+x-1\}$ for $1<x<a_n$. 
		
		\item $\Ap(\Gamma_n,ka_{n}) = a_n \Ap(\Gamma_{n-1},k) \cup (\{a_n b_n,a_n b_n+1,\ldots,a_n b_n+ka_n-1\}\setminus \{a_n b_n,a_n (b_n+1),\ldots, a_n (b_n+k-1)\})$, provided $k>0$ and $ka_n<A_1$. 
		
		\item Consider $k>0$ and $ka_n<A_1$. If $a_n b_n-ka_n\in\Gamma_n$, then 
		\[
		\Ap(\Gamma_n,ka_{n}+1) = \Lambda^1 \cup \cdots \cup \Lambda^n \cup \left( 1+(\Ap(\Gamma_n,ka_n)\cap\Lambda^{n+1}) \right).
		\]
		Otherwise 
		\[
		\Ap(\Gamma_n,ka_{n}+1) = \Lambda^1 \cup \cdots \cup \Lambda^n \cup \left( 1+(\Ap(\Gamma_n,ka_n)\cap\Lambda^{n+1}) \right) \cup \{a_n b_n +1\}.
		\]
		
		\item $\Ap(\Gamma_n,ka_{n}+x) = \Ap(\Gamma_n,ka_{n}+1) \cup \{a_n b_n+ka_n+1,\ldots,a_n b_n+ka_n+x-1\}$ for $1<x<a_n$, provided $k>0$ and $ka_n<A_1$. 
		
		\item $\Ap(\Gamma_n,A_1) = a_n \Ap(\Gamma_{n-1},A_{1}^{(n-1)}) \cup (\{a_n b_n,a_n b_n+1,\ldots,a_n b_n+A_{1}-1\}\setminus \{a_n b_n,a_n (b_n+1),\ldots, a_n (b_n+A_{1}^{(n-1)}-1)\})$.
	\end{enumerate}
	
\end{proposition}

\begin{proof}
	
	Part (1) is trivial, parts (2) and (5) follow from the proof of Lemma \ref{linearity}. Parts (3) and (6) follow from the proof of Lemma \ref{induction}, and part (4) follows from the proof of Lemma \ref{monotony}. 
\end{proof}

By combining in a systematic way the above two propositions, we get an algorithm to compute $\Ap(\Gamma_n,x)$ for every $x\in\{1,\ldots,A_1\}$.  It is well known that the genus of a numerical semigroup can be computed from the Ap\'ery set of any of its elements. We show next how to apply this for inductive numerical semigroups.

\begin{corollary}
	The genus of $\Gamma_n$ equals $\sum_{i=1}^n b_i(a_i-1)$.
\end{corollary}
\begin{proof}
	We use induction on $n$. For $n=0$, the result follows trivially. Assume that the statement holds for $n$ and let us prove it for $n+1$. 
	
	Denote by $g_i$ the genus of $\Gamma_i$, $i\in\{1,\ldots,n\}$. By using Selmer's formula for the genus (\cite[Proposition 2.12]{NS}; $A_1\in \Gamma_n^*$),
	\[ g_n = \frac{1}{A_1}\left(\sum_{w\in \Ap(\Gamma_n,A_1)} w\right)-\frac{A_1-1}2.\]
	By Proposition \ref{Apery-n} (6), we can split $\Ap(\Gamma_n, A_1)$ into $a_n\Ap\left(\Gamma_{n-1}, A_1^{(n-1)}\right)$ and $(\{a_n b_n,a_n b_n+1,\ldots,a_n b_n+A_{1}-1\}\setminus \{a_n b_n,a_n (b_n+1),\ldots, a_n (b_n+A_{1}^{(n-1)}-1)\})$. The sum of the elements in the first set divided by $A_1$ is precisely $g_{n-1}+(A_1^{(n-1)}-1)/2$ (Selmer's formula for $\Gamma_{n-1}$). Now  
	\[
	\frac{1}{A_1}\sum_{i=0}^{A_1-1}(a_nb_n+i)= a_nb_n+\frac{A_1-1}2
	\]
	and 
	\[
	\frac{1}{A_1}\sum_{i=0}^{A_1^{(n-1)}-1}a_n(b_n+i)= \frac{1}{A_1^{(n-1)}}\sum_{i=0}^{A_1^{(n-1)}-1}(b_n+i) =
	b_n+\frac{A_1^{(n-1)}-1}2.
	\]
	Hence 
	\[
	g_n=g_{n-1}+\frac{A_1^{(n-1)}-1}2 -\frac{A_1-1}2 + a_nb_n+\frac{A_1-1}2 - b_n-\frac{A_1^{(n-1)}-1}2 = g_{n-1}+b_n(a_n-1). 
	\]
	Now by using the induction hypothesis, we obtain the desired result.
\end{proof}

\subsection{Towers of function fields}

We study now the particular case of semigroups coming from asymptotically good towers of function fields. In this setting the sequence $a_i$ is constant.

\begin{corollary}\label{particular}

Let $\Gamma_n$ be an inductive numerical semigroup with $a_n=a\geq 2$ for all $n\geq 1$. Then 
\[
\mathrm E(\Gamma_n,2)= \min\left\{1+\sum_{i=1}^{n}\lambda_{i},a+\sum_{i=1}^{n-1}\lambda_{i},\ldots,a^{k}+\sum_{i=1}^{n-k}\lambda_{i},\ldots,a^{n-2}+\sum_{i=1}^{2}\lambda_{i},a^{n-1}+\lambda_{1},a^{n}\right\}.
\]
In particular, if also $\lambda_n=\lambda$ for all $n\geq 1$, then 
\[
\mathrm E(\Gamma_n,2)=\min\{1+n\lambda,a+(n-1)\lambda,\ldots,a^{k}+(n-k)\lambda,\ldots,a^{n-2}+2\lambda,a^{n-1}+\lambda,a^{n}\}.
\]

\end{corollary}

Let us  consider the tower of function fields $({\mathcal T}_{n})$ over $\mathbb{F}_{q^2}$, 
where ${\mathcal T}_{1}=\mathbb{F}_{q^{2}}(x_{1})$ and for $n\geq 2$, ${\mathcal T}_{n}$ is obtained from ${\mathcal T}_{n-1}$ 
by adjoining a new element $x_{n}$ satisfying the equation
\[
x_{n}^{q}+x_{n}=\frac{x_{n-1}^{q}}{x_{n-1}^{q-1}+1}.
\]
This tower was introduced by  Garc\'{\i}a and Stichtenoth in \cite{GS}, and it attains the Drinfeld-Vl\u{a}du\c{t} bound. Thus, error-correcting AG codes coming from this tower have great interest because of its asymptotically good behaviour.  

Let $Q_{n}$ be the rational place on ${\mathcal T}_{n}$ that is the unique pole of $x_{1}\,$. 
It is known that the Weierstrass semigroups $\Gamma_{n}$ of ${\mathcal T}_{n}$ at $Q_{n}$ are as follows: $\Gamma_1=\mathbb{N}$, and for $n\geq 2$,
\[
\Gamma_n = q \cdot \Gamma_{n-1} \cup \{ m\in\mathbb{N} \mid m\geq c_n\},
\]
where 
\[
c_{n}=\left\{\begin{array}{ll}
q^{n}-q^{\frac{n+1}{2}} & \mbox{if $n$ is odd}, \\
q^{n}-q^{\frac{n}{2}} & \mbox{if $n$ is even}. \end{array}\right. 
\]
Thus, the numerical semigroups $\Gamma_n$ are inductive (up to a change of indices, since $n\geq 1$ in this case). We can then apply the previous formulas, assuming $a_1=1$ and $\lambda_1=0$, as follows: first note that $a_{n}=q$ for all $n\geq 2$, and 
\[
b_{n}=\frac{c_{n}}{a_{n}}=\left\{
\begin{array}{ll}
q^{n-1}-q^{\frac{n-1}{2}}&\mbox{if $n$ is odd},\\
q^{n-1}-q^{\frac{n-2}{2}}&\mbox{if $n$ is even},
\end{array}
\right.
\]
so that $\lambda_2=b_2=q-1$. For $n\neq 3$, we have 
\[
\lambda_{n}=b_{n}-c_{n-1}=\left\{
\begin{array}{ll}
0&\mbox{if $n$ is odd},\\
(q-1)q^{\frac{n-2}{2}}&\mbox{if $n$ is even}.
\end{array}
\right.
\]
Notice that the formula still holds true for $n=2$. As a consequence, we have the following result. 

\begin{lemma}\label{tower}
With the notations above, let $\Gamma_n$ ($n\geq 2$)  be the Weierstrass semigroup of the function field ${\mathcal T}_{n}$ at $Q_{n}$. 
Write $n=2m+b$ with $b\in\{0,1\}$. 
\begin{enumerate}[(1)]
\item $A_{n-k}=q^{k+1}$, for $0\leq k\leq n-2$. 

\item $\sharp S_{q^{n-1}}=q^{n-1}$. 

\item $\sharp S_{1}=q^m=q^{\lfloor \frac{n}{2}\rfloor}$. 

\item If $n=2m$, then for $i\in\{1,\ldots,n-2\}$, $\sharp S_{q^{i}}=(q^{\lfloor m-\frac{i}{2}\rfloor}-1)+q^{i}$. 

\item If $n=2m+1$, then for $i\in\{1,\ldots,n-2\}$, $\sharp S_{q^{i}}=(q^{\lceil m-\frac{i}{2}\rceil}-1)+q^{i}$. 
\end{enumerate}
\end{lemma}
\begin{proof}

(1) follows from the definition of $A_j$ and that $a_1=1$.  

(2) follows from the fact that $q^{n-1}$ is the multiplicity of $\Gamma_n$ and thus it belongs to $\Gamma_n$.

In order to prove (3), we apply the formula $\sharp S_{1} = \lambda_{1}+\cdots+\lambda_{n} + 1$ from Theorem \ref{SAi}. Thus, from the above formulas for $\lambda_i$ we get 
\[
\sharp S_{1} = 1 + (q-1) + (q-1)q + \cdots + (q-1)q^{m-1} = 1 + (q-1)(1+q+\cdots+q^{m-1}) = q^{m}.
\]

For (4) and (5), we apply the formula $\sharp S_{A_{n-k}} = \lambda_{1}+\cdots+\lambda_{n-k-1} + A_{n-k}$ also from Theorem \ref{SAi}. 
Note first that $a_1=1$ and $\lambda_1=0$, so that $A_{n-1}=A_{n}$, and $\sharp S_1$ is not relevant when applying Theorem \ref{E2}. 

Thus, if $n=2m$ is even and $i=2j>0$ is even, then 
\begin{multline*}
\sharp S_{A_{n-i+1}} = \sharp S_{q^{i}} = \lambda_{1}+\cdots+\lambda_{n-i} + A_{n-i+1} = [(q-1)(1+q+\cdots+q^{\frac{n-i-2}{2}})] + q^{i}\\ 
=[(q-1)(1+q+\cdots+q^{m-j-1})] + q^{i} = (q^{m-j} - 1) + q^{i} = (q^{m-\frac{i}{2}} - 1) + q^{i} = (q^{\lfloor m-\frac{i}{2}\rfloor} - 1) + q^{i}.
\end{multline*}
Analogously, if $n=2m$ is even and $i=2j+1>0$ is odd, then 
\begin{multline*}
\sharp S_{A_{n-i+1}} = \sharp S_{q^{i}} = \lambda_{1}+\cdots+\lambda_{n-i} + A_{n-i+1} = [(q-1)(1+q+\cdots+q^{\frac{n-i-3}{2}})] + q^{i}\\
=[(q-1)(1+q+\cdots+q^{m-j-2})] + q^{i} = (q^{m-j-1} - 1) + q^{i} = (q^{m-\frac{i+1}{2}} - 1) + q^{i} = (q^{\lfloor m-\frac{i}{2}\rfloor} - 1) + q^{i}.
\end{multline*}
The calculations for $n=2m+1$ are similar. In fact, if $n=2m+1$ is odd and $i=2j>0$ is even, then 
\begin{multline*}
\sharp S_{A_{n-i+1}} = \sharp S_{q^{i}} = \lambda_{1}+\cdots+\lambda_{n-i} + A_{n-i+1} = [(q-1)(1+q+\cdots+q^{\frac{n-i-3}{2}})] + q^{i} \\
=[(q-1)(1+q+\cdots+q^{m-j-1})] + q^{i} = (q^{m-j} - 1) + q^{i} = (q^{m-\frac{i}{2}} - 1) + q^{i} = (q^{\lceil m-\frac{i}{2}\rceil} - 1) + q^{i}.
\end{multline*}
Finally, if $n=2m+1$ is odd and $i=2j+1>0$ is odd,  we obtain 
\begin{multline*}
\sharp S_{A_{n-i+1}} = \sharp S_{q^{i}} = \lambda_{1}+\cdots+\lambda_{n-i} + A_{n-i+1} = [(q-1)(1+q+\cdots+q^{\frac{n-i-2}{2}})] + q^{i} \\
=[(q-1)(1+q+\cdots+q^{m-j-1})] + q^{i} = (q^{m-j} - 1) + q^{i} = (q^{m-\frac{i-1}{2}} - 1) + q^{i} = (q^{\lceil m-\frac{i}{2}\rceil} - 1) + q^{i}.\qedhere
\end{multline*}

\end{proof}

Note that $\lambda_i=0$ if $i$ is odd, and thus we can skip half of the values in Theorem \ref{SAi}, in order to get the minimum in Theorem \ref{E2}. In fact, by looking at the formulas (4) and (5) in the above lemma, we need just to consider $\sharp S_{q^{i}}$ with $i$ even if $n$ is odd, and $\sharp S_{q^{i}}$ for $i$ odd if $n$ is even. Thus one has the following result. 

\begin{proposition}\label{reduction}
Under the standing notation and hypothesis, the second Feng-Rao number of the Weierstrass semigroup $\Gamma_n$ of the function field ${\mathcal T}_{n}$ at $Q_{n}$ is given by the minimum of the following numbers: 
\[
\begin{array}{l}
\sharp S_{1}=q^{\lfloor \frac{n}{2}\rfloor},\\
\sharp S_{q^{n-1}}=q^{n-1},\\
\sharp S_{q^{n-1-2k}}=(q^{k}-1)+q^{n-1-2k},\mbox{ for }k\in\{1,\ldots,\lfloor \frac{n}{2}\rfloor-1\}.
\end{array}
\]
\end{proposition}

\begin{proof}
Note that, according the formulas given in Lemma \ref{tower}, the missing numbers are not relevant for the minimum, since 
\[
\sharp S_{q^{n-1-2k}}=(q^{k}-1)+q^{n-1-2k} < (q^{k}-1)+q^{n-2k} = \sharp S_{q^{n-2k}}
\]
if $n$ is even; and if $n=2m+1$ is odd, then 
\[
\sharp S_{1} = q^{\lfloor \frac{n}{2}\rfloor} = q^{m} \leq (q^{m} - 1) + q = \sharp S_{q}
\]
(note that in both cases $q>1$). 
\end{proof}

Now we are able to compute explicitly the above minimum, and get the second Feng-Rao number as a consequence. First we need a technical lemma. 

\begin{lemma}\label{technical}

Assume that $q\geq 2$, $0\leq i\leq j$ and $0\leq k\leq l$. 
\begin{enumerate}[(1)]
\item If $j=l$ and $i<k$, then $q^{i}+q^{j}<q^{k}+q^{l}$. 

\item If $j<l$, then $q^{i}+q^{j}<q^{k}+q^{l}$. 
\end{enumerate}
\end{lemma}
\begin{proof}
The first case is obvious. For the second case note that, since $q\geq 2$ and $q^{l}>q^{j}\geq q^{i}$, one has 
$
q^{j} + q^{i} \leq q^{l}
$
and $q^{k}>0$. 
\end{proof}

\begin{theorem}\label{FR2}
Let $\Gamma_n$ be the Weierstrass semigroup of the function field ${\mathcal T}_{n}$ at $Q_{n}$. 

\begin{enumerate}[(1)]

\item $\mathrm E(\Gamma_1,2)=1$. 

\item $\mathrm E(\Gamma_2,2)=\mathrm E(\Gamma_3,2)=q$. 

\item $\mathrm E(\Gamma_4,2)=2q-1$. 

\item $\mathrm E(\Gamma_5,2)=q^{2}$. 

\item For $n\geq 6$, $\mathrm E(\Gamma_n,2)=q^{\lceil\frac{n-1}{3}\rceil}+q^{n-1-2\lceil\frac{n-1}{3}\rceil}-1$. 

\end{enumerate}

\end{theorem}

\begin{proof}

First note that for $n\geq 3$ one has $\sharp S_{1}=q^{\lfloor \frac{n}{2}\rfloor}<\sharp S_{q^{n-1}}=q^{n-1}$ (Proposition \ref{reduction}), 
so that the minimum in Proposition \ref{reduction} cannot be achieved in $\sharp S_{q^{n-1}}$. Next we distinguish all the cases of the statement and will use Proposition \ref{reduction} without referring to it.

\begin{enumerate}[(1)]
\item The semigroup $\Gamma_1=\mathbb{N}$ has genus $g=0$, and hence $\mathrm E(\Gamma_1,2)=1$ (see \cite{WCC}; also here the multiplicity is $1$, and the cardinality of the Ap\'ery set 1 in $\mathbb N$ is 1). 

\item For $\Gamma_2$, $\sharp S_1=\sharp S_q=q$, and hence $\mathrm E(\Gamma_2,2)=q$.

\item For $\Gamma_3$,  $\sharp S_1=q$. 
Hence $\mathrm E(\Gamma_3,2)=q$. 

\item For $\Gamma_4$, we obtain $\sharp S_1=q$, $\sharp S_q=2q-1$; 
whence $\mathrm E(\Gamma_4,2)=2q-1$. 

\item For $\Gamma_5$, we have $S_1=q^{2}$ and 
$S_q=S_{q^{2}}=q^{2}+q-1$. 
In this case, $\mathrm E(\Gamma_5,2)=q^{2}$.

\item For $n\geq 6$, we claim that the minimum in Proposition \ref{reduction} is achieved is a number of the form 
\[
\sharp S_{q^{n-1-2k}}=(q^{k}-1)+q^{n-1-2k}
\]
for some $k\in\{1,\ldots,\lfloor \frac{n}{2}\rfloor-1\}$. In other words, the minimum is not achieved in $\sharp S_1=q^{\lfloor \frac{n}{2}\rfloor}$. The set $\{\lceil(n-1-\lfloor\frac{n}2\rfloor)/2\rceil,\ldots, \lfloor \frac{n}2\rfloor-1\}$ is not empty for $n\ge 5$. Take $k$ in this set. Then $n-1-2k<\lfloor \frac{n}2\rfloor$ and  
\[
\sharp S_{q^{n-1-2k}}=(q^{k}-1)+q^{n-1-2k}\leq q^{\lfloor\frac{n}{2}\rfloor-1} + q^{\lfloor\frac{n}{2}\rfloor-1} -1 \leq q^{\lfloor\frac{n}{2}\rfloor} - 1 < \sharp S_1.
\]
Thus, all we have to do, after cancelling the $-1$, is to compare possible sums of two powers of $q$ of the form $q^{k}+q^{n-1-2k}$, with $k$ ranging in the above set; and here is when technical Lemma \ref{technical} becomes useful (in fact, we just need part (2) of this lemma). 

We distinguish to cases.

\begin{itemize}
\item If $k\ge n-1-2k$, then the minimum $q^{n-1-2k}+q^k$ is achieved in the least integer $k$ fulfilling the inequality $k\ge n-1-2k$, which is $\lceil \frac{n-1}3\rceil$.

\item If $k\le n-1-2k$, then the minimum $q^{k}+q^{n-1-2k}$ is reached in the largest integer $k$ such that $k\le n-1-2k$. This integer is $\lfloor \frac{n-1}3\rfloor$.
\end{itemize}
So it only remains to compare $q^{n-1-2\lceil\frac{n-1}{3}\rceil}+q^{\lceil\frac{n-1}{3}\rceil}$ with $q^{\lfloor\frac{n-1}{3}\rfloor}+ q^{n-1-2\lfloor\frac{n-1}{3}\rfloor}$. In light of Lemma \ref{technical} again, we only need to compare $\lceil\frac{n-1}{3}\rceil$ with $n-1-2\lfloor\frac{n-1}{3}\rfloor$. By checking with the three different residues of $n-1$ modulo 3, one easily proves that $2\lfloor \frac{n-1}3\rfloor +\lceil \frac{n-1}2\rceil \le n-1$, and thus $\lceil\frac{n-1}{3}\rceil \le  n-1-2\lfloor\frac{n-1}{3}\rfloor$.\qedhere

\end{enumerate}
\end{proof}

\section{Patterns and inductive semigroups}\label{sec:patterns}

If we consider the operation $\Gamma'=a\Gamma\cup(a b+\mathbb N)$, with $a$ and $b$ positive integers with $a\ge 2$ and $b$ greater than the Frobenius number of $\Gamma$, it is known that $\Gamma'$ is Arf if $\Gamma$ is Arf (\cite{Arf}).  In particular, every inductive semigroup $\Gamma_n$ has maximal embedding dimension (the embedding dimension coincides with the multiplicity; see for instance \cite[Chapter 2]{NS}). And a minimal generating system for $\Gamma_n$ is $(\Ap(\Gamma_n, A_1)\setminus\{0\})\cup\{A_1\}$ (\cite[Proposition 3.1]{NS}); this set is described in Proposition \ref{Apery-n}(6).

Now some natural question arise. 

\begin{itemize}

\item Is $\Gamma'$ saturated if $\Gamma$ is saturated? 

\item What kind of patterns are preserved from $\Gamma$ to $\Gamma'$? 

\end{itemize}

A numerical semigroup $\Gamma$ is \emph{saturated} if for every $s\in \Gamma$ and every $s_1,\ldots, s_n\in \Gamma$ with $n\in \mathbb N$ the element $s+z_1s_1+\cdots +z_ns_n$ is again in $\Gamma$ for all $z_1,\ldots, z_n\in \mathbb Z$ with $z_1s_1+\cdots +z_ns_n\ge 0$. It can be shown (see for instance \cite[Proposition 3.34]{NS}) that $\Gamma$ is saturated if and only if for all $s\in \Gamma\setminus\{0\}$, $s+\gcd(\Gamma\cap [0,s])\in \Gamma$. We use this last characterization to prove that the saturated property is also preserved by multiplication.

\begin{proposition}\label{pres-sat}
Let $\Gamma$ be a saturated numerical semigroup, and let $a$ and $b$ be positive integers. Then $\bar{\Gamma}=a\Gamma\cup (b+\mathbb N)$ is also saturated.
\end{proposition}
\begin{proof}
Let $s\in \bar{\Gamma}$. 
If $s\ge b$, then trivially the condition $s+\gcd(\bar{\Gamma}\cap [0,s])\in \bar{\Gamma}$ is fulfilled. So assume that $s<b$. In this case $s\in a\Gamma$, and there exists $t\in \Gamma$ such that $s=at$. Also $\bar{\Gamma}\cap [0,s]=a(\Gamma\cap[0,t])$, and consequently $s+\gcd(\bar{\Gamma}\cap [0,s])=a(t+\gcd(\Gamma\cap [0,t]))\in a\Gamma\subset \bar{\Gamma}$, because $t+\gcd(\Gamma\cap[0,t])\in\Gamma$. 
\end{proof}

\begin{example}\label{ex-q}
It may happen that $a\Gamma\cup (b+\mathbb N)$ is saturated without being $\Gamma$ saturated.

\begin{verbatim}
gap> s:=NumericalSemigroup(3,5);  
<Modular numerical semigroup satisfying 10x mod 15 <= x >
gap> IsSaturatedNumericalSemigroup(s);
false
gap> t:=MultipleOfNumericalSemigroup(s,5,11);
<Numerical semigroup>
gap> IsSaturatedNumericalSemigroup(t);
true
\end{verbatim}

The above example was found with the aid of the computer algebra system \GAP \ and the package {\tt NumericalSgps} (see~\cite{numericalsgps} and \cite{GAP4}). 

\end{example}

The behavior in Example \ref{ex-q} is a consequence of the fact that we are not imposing any condition on $b$, and thus $\Gamma'=a\Gamma\cup(b+\mathbb N)$ for infinitely many numerical semigroups $\Gamma$. This is probably why the authors in \cite{ds} include the condition $b> a f$, with $f$ the Frobenius number of $\Gamma$, for proving Proposition \ref{pres-sat}. This condition is compatible with the definition of inductive semigroup: we are forcing $b_{i+1}\ge a_ib_i$, and $b_i$ is precisely the conductor of $\Gamma_i$.

The condition $b\ge ac$, with $c$ the conductor of $\Gamma$, ensures that $\Gamma$ will be the only numerical semigroup yielding $\Gamma'=a\Gamma\cup (b+\mathbb N)$. In this setting, $\Gamma'/a:=\{x\in \mathbb N\mid ax\in \Gamma'\}$ is precisely $\Gamma$, and thus   \cite[Proposition 2.2]{ds}, asserts that if $\Gamma'$ is saturated, then so is $\Gamma$ (the same for the Arf property: \cite[Proposition 2.3]{ds}).  In this case, we say that $\Gamma'$ is an $a$-\emph{fold} of $\Gamma$.

A \emph{pattern} of length $n$, $n$ a positive integer, is an expression of the form $a_1x_1+\dots +a_nx_n$ with $a_1,\ldots, a_n\in \mathbb Z\setminus\{0\}$ and $x_i$ unknowns. We say that a numerical semigroup $\Gamma$ \emph{admits} the pattern $p$ if for every $s_1,\ldots, s_n\in \Gamma$ with $s_1\ge s_2\ge \dots \ge s_n$, the integer $p(s_1,\ldots,s_n)\in \Gamma$. We say that $p$ is \emph{admissible} if there is at least a numerical semigroup admitting $p$. It follows easily that if $p$ is admissible, then $a_1\ge 0$ (see for instance \cite[Chapter 6, Section 2.4] {NS}). Arf semigroups are precisely the numerical semigroups admitting the pattern $x_1+x_2-x_3$ (or equivalently, $2x_1-x_2$).

For an admissible pattern $p=a_1x_1+\dots +a_nx_n$, define $p'=(a_1-1)x_1+a_2x_2+\dots +a_nx_n$. We say that $p$ is \emph{strongly admissible} if $p'$ is also admissible. It is well known that the set of semigroups admitting a strongly admissible pattern is closed under intersections and the adjoin of the Frobenius number (and thus it is a Frobenius variety; see for instance \cite[Chapter 6]{NS}). We show that this property is also preserved under multiples.

\begin{proposition}
Let $p$ be a strongly admissible pattern. Let $\Gamma$ be a numerical semigroup and $a,b$ be positive integers. Assume that $\Gamma$ admits $p$. Then so does $\bar{\Gamma}=a\Gamma\cup (b+\mathbb N)$.
\end{proposition}
\begin{proof}
Let $p=a_1x_1+\dots+a_nx_n$. It is well known that if $s_1\ge s_2\ge \cdots \ge s_n$, then $p(s_1,\ldots,s_n)\ge s_1$ (\cite[Lemma 7.15]{NS}). Hence whenever $s_1\ge b$, it trivially follows that $p(s_1,\ldots,s_n)\in \bar{\Gamma}$. So assume that $s_1,\ldots,s_n$ are in $\Gamma$ and $b>s_1\ge\cdots \ge s_n$. Then $s_i\in a\Gamma$ for all $i$, and consequently $s_i=at_i$ for some $t_i\in \Gamma$. It follows that $p(s_1,\ldots,s_n)=ap(t_1,\ldots, t_n)$, and as $p(t_1,\ldots, t_n)\in \Gamma$, we deduce that $p(s_1,\ldots, s_n)\in a\Gamma \subset \bar{\Gamma}$.
\end{proof}

Let $p$ be a strongly admissible pattern. Define $\mathcal S(p)$ as the set  of numerical semigroups admitting $p$. We already know that this set has infinitely many elements, as a consequence of \cite[Corollary 15]{patterns}; since this result states that whenever $\Gamma\in \mathcal S(p)$, then so does $\Gamma\setminus\{\mathrm m(\Gamma)\}$. The above proposition also implies this infiniteness condition.

\begin{corollary}
Let $p$ be a strongly admissible pattern. Then $\mathcal S(p)$ contains the class of all inductive numerical semigroups.
\end{corollary}

\section{Conclusions}\label{sec:examples_conclusions}

Theorem \ref{FR2} in Section \ref{sec:main} provides us with an estimate for the second Hamming weight of codes coming from an inductive tower of function fields. We recall now the definition of the generalized (Hamming) weights. First, we define the support of a linear code $C$ as
\[
{\rm supp}(C):=\{i \mid c_{i}\neq 0\;\;\mbox{for some ${\bf c}\in C$}\}.
\]
Thus, the $r$th generalized weight of $C$ is defined by
\[
{\mathrm d}_{r}(C):=\min\{\sharp\,{\rm supp}(C')\mid \mbox{$C'$ is a linear subcode of $C$ with ${\rm dim}(C')=r$}\}.
\]
In fact, the above definition only makes sense if $r\leq k$, where $k$ is the dimension of $C$.

\begin{theorem}\label{thm-final}

Let $\Gamma_n = \{0=\rho_1<\rho_2<\cdots <\rho_i<\cdots\}$ be the inductive Weierstrass semigroup of the function field ${\mathcal T}_{n}$\/, defined over the finite field $\mathbb{F}_{q}$\/, at the rational place $Q_{n}$\/, as in~\cite{GS}, and let $C_{m}$ be a dual one-point AG code given by the divisor $G=mQ_{n}$ (see \cite{HvLP} for the details), with $m\geq c$, $c$ being the conductor of $\Gamma_n$. Then

\begin{enumerate}[(1)]


\item ${\mathrm d}_{2}(C_{m})\geq m+2-2g+q$ for $n\in\{2,3\}$, 

\item ${\mathrm d}_{2}(C_{m})\geq m+1-2g+2q$ for $n=4$,

\item ${\mathrm d}_{2}(C_{m})\geq m+2-2g+q^{2}$ for $n=5$, 

\item ${\mathrm d}_{2}(C_{m})\geq m+2-2g+q^{\lceil\frac{n-1}{3}\rceil}+q^{n-1-2\lceil\frac{n-1}{3}\rceil}-1$, for $n\geq 6$. 

\end{enumerate}
\end{theorem}
\begin{proof}
This is a consequence of the formula ${\mathrm d}_{2}(C_{m})\geq\delta^{2}(m+1)\geq m+2-2g+\mathrm E(\Gamma_n,2)$ from~\cite{WCC}, and Theorem \ref{FR2}. 
\end{proof}

It follows from \cite[Theorem 2.8]{KirPel} that 
\[
{\mathrm d}_{r}(C_{m}) \geq \delta_{FR}(m+r) \geq m+r+1-2g = m+2-2g+(r-1).
\]
Thus, our result improves this bound since 
\[
{\mathrm d}_{r}(C_{m}) \geq \delta_{FR}^{r}(m+1) \geq m+2-2g + E_r,
\]
and actually $E_r \geq r-1$ for every numerical semigroup (see \cite{WCC}).

\begin{remark}\label{IEEE-46}

It was stated in~\cite[Theorem 46]{fr-IEEE} that 
\[
d_{r}(C_m)\geq\delta_{FR}(m+1)+E_r = \delta_{FR}(m+1)+\rho_r
\]
for numerical semigroups with two generators, but the authors missed the condition 
\[
m+1=(c-1)+\rho_k
\]
for every $k\geq 2$, since for those values of $m$ one has $\delta_{FR}^{r}(m+1)=m+2-2g+E_r=\delta_{FR}(m+1)+E_r$ (see~\cite{WCC}). 
Otherwise the inequality may not be true. In fact, this inequality is not true in general for non-symmetric semigroups, like inductive or Arf ones. 
In fact, it is easy to find examples with \GAP \ such that 
\[
\delta_{FR}(m+1)+E_r > \delta_{FR}^{r}(m+1).
\]

\end{remark}

Finally, a generalization of the Griesmer bound for the generalized Hamming weights states that 
\[
{\mathrm d}_{r}(C)\geq\displaystyle\sum_{i=1}^{r-1}\left\lceil\displaystyle\frac{\mathrm d(C)}{q^{i}}\right\rceil,
\]
where $\mathrm d(C)\equiv {\mathrm d}_{1}(C)$ is the minimum distance of the code $C$, which is defined over the finite field $\mathbb{F}_{q}$ (see~\cite{Griesmer}).  In particular, for $r=2$ one has
\[
{\mathrm d}_{2}(C)\geq \mathrm d(C)+\left\lceil\displaystyle\frac{\mathrm d(C)}{q}\right\rceil.
\]
Since we are just using the semigroup for estimating the generalized Hamming weights of the AG code $C_m$, $C_m$ being constructed from an asymptotically good tower function fields as in~\cite{GS}, we have to substitute $\mathrm d(C_{m})$ by the order bound $\delta_{FR}(m+1)$ obtaining, for $r=2$, the bound
\[
{\mathrm d}_{2}(C_m)\geq \mathrm{GOB}(m+1):=\delta_{FR}(m+1)+\left\lceil\displaystyle\frac{\delta_{FR}(m+1)}{q}\right\rceil.
\]

Denote again $\delta_{FR}(m)=\delta(m)$ for simplicity. 
In order to apply the above {\em Griesmer order bound}\/, introduced in~\cite{fr-IEEE}, we need to compute the Feng-Rao distance of inductive semigroups, but this calculation was done in~\cite[Theorem 4.1]{Arf}. Thus, since an inductive semigroup $\Gamma_n$ is Arf, for $m\geq c=\rho_r$ one has: 

\begin{itemize}

\item $\delta(m)=2$ for $c\leq m<m+\rho_2$, 

\item $\delta(m)=2i$ for $c+\rho_i\leq m<\rho_{i+1}$, 

\item $\delta(m)=m+1-2g$ for $m\geq c+\rho_{r}-1=2c-1$. 

\end{itemize}

We use this formulas to compute the following examples.

\begin{example}\label{ex-final-1}

Consider $q=9$. The second inductive semigroup of the tower in~\cite{GS} is 
\[
\Gamma_{2}=\{0,9,18,27,36,45,54,63,72,\rightarrow\}.
\]
In this case $g=64$, $c=72$, and $E_2=9$. Note that these semigroups are not symmetric, and consequently $c<2g$. For the involved AG codes we must consider only $m>2g-2$. Hence we show next the results for the codes with $2g-1\leq m\leq 2c-2$. Note also that for a code $C(m)$ the bounds for the Hamming weights correspond to $m+1$: 

\bigskip 

\begin{tabular}{|c|cccccccccccc|}
\hline 
$m$ & 127 & $\cdots$ & 133 & 134 & 135 & 136 & 137 & 138 & 139 & 140 & 141 & 142 \\
\hline 
\hline 
$m+1-2g+E_2$ & 10 & $\cdots$ & 16 & 17 & 18 & {\bf 19} & {\bf 20} & {\bf 21} & {\bf 22} & {\bf 23} & {\bf 24} & {\bf 25} \\
\hline 
$GOB(m+1)$ & 16 & $\cdots$ & 16 & 18 & 18 & 18 & 18 & 18 & 18 & 18 & 18 & 18 \\
\hline 
\end{tabular}

\end{example}

\begin{example}\label{ex-final-2}

Consider now $q=16$ (note that in~\cite{GS} $q$ must be a square).The second inductive semigroup of the tower is now 
\[
\Gamma_{2}=\{0,16,32,48,64,80,96,112,128,144,160,176,192,208,224, 240,\rightarrow\}.
\]
In this case $g=225$, $c=240$, and $E_2=16$. The results for the codes with $2g-1\leq m\leq 2c-2$ are now: 


\bigskip 

\begin{center}
\begin{tabular}{|c|cccccccccc|}
\hline 
$m$ &  449 & $\cdots$ & 462 & 463 & 464 & 465 & 466 & 467 & 468 & 469 \\
\hline 
\hline 
$m+1-2g+E_2$ & 17 & $\cdots$ & 30 & 31 & 32 & {\bf 33} & {\bf 34} & {\bf 35} & {\bf 36} & {\bf 37} \\
\hline 
\hline 
$GOB(m+1)$ & 30 & $\cdots$ & 30 & 32 & 32 & 32 & 32 & 32 & 32 & 32  \\
\hline 
\hline
$m$ & 470 & 471 & 472 & 473 & 474 & 475 & 476 & 477 & 478 & \\
\hline 
\hline 
$m+1-2g+E_2$ & {\bf 38} & {\bf 39} & {\bf 40} & {\bf 41} & {\bf 42} & {\bf 43} & {\bf 44} & {\bf 45} & {\bf 46} & \\
\hline 
$GOB(m+1)$  & 32 & 32 & 32 & 32 & 32 & 32 & 32 & 32 & 32 & \\
\hline 
\end{tabular}
\end{center}

\end{example}

In both examples, the Goppa-like bounds improving the Griesmer-like one are emphasized in bold.

\end{document}